\documentclass[10pt]{article}
\usepackage[utf8]{inputenc}
\usepackage[margin=1in]{geometry}
\usepackage{rotating}
\usepackage{booktabs}
\usepackage{array}
\usepackage{amsmath, amsfonts}
\usepackage{color}
\usepackage{hyperref}
\usepackage{bbm}
\usepackage{algorithm}
\usepackage[noend]{algpseudocode}
\usepackage{subcaption}
\usepackage{enumerate}
\usepackage{stmaryrd,scalerel}
\usepackage{mdframed}
\usepackage{enumitem}

\usepackage{mathtools}
\mathtoolsset{showonlyrefs}

\usepackage[
    backend=biber,
    style=ieee,
    citestyle=numeric-comp,
    sorting=none,
    giveninits=true,
    natbib,
    hyperref,
    maxbibnames=99,
    doi=false,isbn=false,url=false,eprint=false
]{biblatex}

\usepackage{amsthm}
\usepackage[normalem]{ulem}
\usepackage{soul}

\newtheorem{theorem}{Theorem} 
\newtheorem{lemma}{Lemma} 
\newtheorem{prop}{Proposition}
\newtheorem{cor}{Corollary}

\usepackage{chngcntr}
\usepackage{apptools}
\AtAppendix{\counterwithin{theorem}{section}}
\AtAppendix{\counterwithin{prop}{section}}

\usepackage{mathtools}

\usepackage[dvipsnames]{xcolor}

\def\E{\mathbb{E}}

\newcommand{\X}{\mathcal{X}}
\newcommand{\C}{\mathcal{C}}

\renewcommand{\S}{\mathcal{S}}

\newcommand{\lhat}{\hat{\lambda}}

\newcommand{\Rhat}{\widehat{R}}

\newcommand{\ind}[1]{\mathbbm{1}\left\{#1\right\}}

\newcommand{\triangleq}{\stackrel{\triangle}{=}}

\renewcommand{\P}{\mathbb{P}}
\newcommand{\rb}{\right)}
\newcommand{\lb}{\left(}
\newcommand{\eps}{\epsilon}

\newcommand{\lhatplus}{\lhat^{\uparrow}}
\newcommand{\Rplus}{R^{\uparrow}}
\newcommand{\Rhatplus}{\Rhat^{\uparrow}}
\newcommand{\Chat}{\widehat{\mathcal{C}}}
\newcommand{\Clam}{\mathcal{C}_{\lambda}}
\newcommand{\Clhat}{\mathcal{C}_{\hat{\lambda}}}

\newcommand{\D}{\mathcal{D}}

\makeatletter
\def\blfootnote{\xdef\@thefnmark{}\@footnotetext}
\makeatother

\title{\vspace{-1cm}Conformal Risk Control}
\author{Anastasios N.~Angelopoulos$^1$, Stephen Bates$^1$, Adam Fisch$^2$, Lihua Lei$^3$, and Tal Schuster$^4$}
\date{
$^1$University of California, Berkeley\\
$^2$Massachusetts Institute of Technology\\
$^3$Stanford University\\
$^4$Google Research\\
}

\begin{document}
\maketitle

\begin{abstract}
    We extend conformal prediction to control the expected value of any monotone loss function.
    The algorithm generalizes split conformal prediction together with its coverage guarantee.
    Like conformal prediction, the conformal risk control procedure is tight up to an $\mathcal{O}(1/n)$ factor.
    We also introduce extensions of the idea to distribution shift, quantile risk control, multiple and adversarial risk control, and expectations of U-statistics.
    Worked examples from computer vision and natural language processing demonstrate the usage of our algorithm to bound the false negative rate, graph distance, and token-level F1-score.
\end{abstract}

\section{Introduction}
\label{sec:intro}
We seek to endow some pre-trained machine learning model with guarantees on its performance as to ensure its safe deployment.
Suppose we have a base model $f$ that is a function mapping inputs $x \in \X$ to values in some other space, such as a probability distribution over classes. 
Our job is to design a procedure that takes the output of $f$ and post-processes it into quantities with desirable statistical guarantees.

Split conformal prediction~\citep{vovk2005algorithmic,papadopoulos2002inductive}, which we will henceforth refer to simply as “conformal prediction”, has been useful in areas such as computer vision~\citep{angelopoulos2020sets} and natural language processing~\citep{fisch2020efficient} to provide such a guarantee.
By measuring the model's performance on a \emph{calibration dataset} $\big\{(X_{i},Y_{i})\big\}_{i=1}^{n}$ of feature-response pairs, conformal prediction post-processes the model to construct prediction sets that bound the \emph{miscoverage},\looseness=-1
\begin{equation}
    \label{eq:miscoverage}
    \P\Big( Y_{n+1} \notin \C(X_{n+1}) \Big) \leq \alpha,
\end{equation}
where $(X_{n+1},Y_{n+1})$ is a new test point, $\alpha$ is a user-specified error rate (e.g., 10\%), and $\C$ is a function of the model and calibration data that outputs a prediction set.
Note that $\C$ is formed using the first $n$ data points, and the probability in \eqref{eq:miscoverage} is over the randomness in all $n+1$ data points (i.e., the draw of both the calibration and test points).

In this work, we extend conformal prediction to prediction tasks where the natural notion of error is not simply miscoverage. 
In particular,
our main result is that a generalization of conformal prediction provides guarantees of the form
\begin{equation}
    \label{eq:risk-upper-bound}
    \E\Big[ \ell\big(\C(X_{n+1}), Y_{n+1}\big) \Big] \leq \alpha,
\end{equation}
for any bounded \emph{loss function} $\ell$ that shrinks as $\C(X_{n+1})$ grows.
We call this a \emph{conformal risk control} guarantee.
Note that~\eqref{eq:risk-upper-bound} recovers the conformal miscoverage guarantee in~\eqref{eq:miscoverage} when using the miscoverage loss, $\ell\big(\C(X_{n+1}), Y_{n+1}) = \ind{Y_{n+1} \notin \C(X_{n+1})}$.
However, our algorithm also extends conformal prediction to situations where other loss functions, such as the false negative rate (FNR) or F1-score, are more appropriate.

As an example, consider multilabel classification, where the $Y_i \subseteq \{1,...,K\}$ are sets comprising a subset of $K$ classes.
Given a trained multilabel classifier $f : \X \to [0,1]^K$, we want to output sets that include a large fraction of the true classes in $Y_i$. 
To that end, we post-process the model's raw outputs into the set of classes with sufficiently high scores, $\Clam(x) = \{ k : f(X)_k \geq 1- \lambda \}$.
Note that as the threshold $\lambda$ grows, we include more classes in $\Clam(x)$---i.e., it becomes more conservative.
In this case, conformal risk control finds a threshold value $\lhat$ that controls the fraction of missed classes, i.e., the expected value of $\ell\big( \Clhat(X_{n+1}), Y_{n+1} \big) = 1 - |Y_{n+1} \cap \Clam(X_{n+1})|/ |Y_{n+1}|$.
Setting $\alpha=0.1$ would ensure that our algorithm produces sets $\Clhat(X_{n+1})$ containing $\geq90\%$ of the true classes in $Y_{n+1}$ on average.\looseness=-1

\subsection{Algorithm and preview of main results}
Formally, we will consider post-processing the predictions of the model $f$ to create a function $\Clam(\cdot)$. 
The function has a parameter $\lambda$ that
encodes its level of conservativeness: larger $\lambda$ values yield more conservative outputs (e.g., larger prediction sets). 
To measure the quality of the output of $\Clam$, we consider a loss function $\ell(\Clam(x), y) \in (-\infty, B]$ for some $B<\infty$.
We require the loss function to be non-increasing as a function of $\lambda$.
Our goal is to choose $\lhat$ based on the observed data $\big\{(X_{i},Y_{i})\big\}_{i=1}^{n}$ so that risk control as in~\eqref{eq:risk-upper-bound} holds.

We now rewrite this same task in a more notationally convenient and abstract form.
Consider an exchangeable collection of non-increasing, random functions $L_i : \Lambda \to (-\infty, B]$, $i=1,\dots,n+1$. Throughout the paper, we assume $\lambda_{\max}\triangleq \sup \Lambda \in \Lambda$. We seek to use the first $n$ functions to choose a value of the parameter, $\lhat$, in such a way that the risk on the unseen function is controlled:
\begin{equation}
  \label{eq:intro-risk-control}
  \E\Big[L_{n+1}\big(\lhat \big)\Big] \leq \alpha.
\end{equation}
We are primarily motivated by the case where $L_i(\lambda) = \ell(\Clam(X_i), Y_i)$, in which case the guarantee in~\eqref{eq:intro-risk-control} coincides with risk control as in~\eqref{eq:risk-upper-bound}.

Now we describe the algorithm.
Let $\Rhat_{n}(\lambda) = (L_1(\lambda) + \ldots + L_n(\lambda))/n$.
Given any desired risk level upper bound $\alpha \in (-\infty, B)$, define \looseness=-1
\begin{equation}
\label{eq:lhat}
\lhat = \inf\left\{ \lambda : \frac{n}{n+1} \Rhat_{n}(\lambda)  + \frac{B}{n+1} \leq \alpha \right\}.
\end{equation}
When the set is empty, we define $\hat{\lambda} = \lambda_{\max}$.
Our proposed \emph{conformal risk control} algorithm is to deploy $\lhat$ on the forthcoming test point. 
Our main result is that this algorithm satisfies~\eqref{eq:intro-risk-control}.
When the $L_i$ are i.i.d.\ from a continuous distribution, the algorithm satisfies a lower bound saying it is not too conservative,
\begin{equation}
  \E\Big[ L_{n+1}\big(\lhat \big) \Big] \geq \alpha - \frac{2B}{n+1}. 
\end{equation}
We show the reduction from conformal risk control to conformal prediction in Section~\ref{sec:cp_is_crc}.
Furthermore, if the risk is non-monotone, then this algorithm does not control the risk; we discuss this in Section~\ref{sec:counterexample}. 
Finally, we provide both practical examples using real-world data and several theoretical extensions of our procedure in Sections~\ref{sec:examples} and \ref{sec:extensions}, respectively.

\subsection{Related work}
Conformal prediction was developed by Vladimir Vovk and collaborators beginning in the late 1990s~\citep{vovk1999machine, vovk2005algorithmic}, and has recently become a popular uncertainty estimation tool in the machine learning community, due to its favorable model-agnostic, distribution-free, finite-sample guarantees.
See~\cite{angelopoulos-gentle} for a modern introduction to the area or~\cite{shafer2008tutorial} for a more classical alternative.
As previously discussed, in this paper we primarily build on \emph{split conformal prediction}~\citep{papadopoulos2002inductive}; statistical properties of this algorithm including the coverage upper bound were studied in~\cite{lei2018distribution}.
Recently there have been many extensions of the conformal algorithm, mainly targeting deviations from exchangeability~\citep{tibshirani2019conformal,gibbs2021adaptive,barber2022conformal,fannjiang2022conformal} and improved conditional coverage~\citep{barber2019limits,romano2019conformalized,guan2020conformal,romano2020classification,angelopoulos2020sets}. 
Most relevant to us is recent work on risk control in high probability~\citep{vovk2012conditional, bates2021distribution,angelopoulos2021learn} and its applications~\citep[][\emph{inter alia}]{Park2020PAC,fisch2022conformal, schuster2021consistent, schuster2022confident, sankaranarayanan2022semantic, angelopoulos2022image, angelopoulos2022recommendation}.
Though these works closely relate to ours in terms of motivation, the algorithm presented herein differs greatly: it has a guarantee in expectation, and neither the algorithm nor its analysis share much technical similarity with these previous works.

To elaborate on the difference between our work and previous literature, first consider conformal prediction.
The purpose of conformal prediction is to provide coverage guarantees of the form in~\eqref{eq:miscoverage}.
The guarantee available through conformal risk control,~\eqref{eq:intro-risk-control}, strictly subsumes that of conformal prediction; it is generally impossible to recast risk control as coverage control.
As a second question, one might ask whether~\eqref{eq:intro-risk-control} can be achieved through standard statistical machinery, such as uniform concentration inequalities.
Though it is possible to integrate a uniform concentration inequality to get a bound in expectation, this strategy tends to be excessively loose both in theory and in practice (see, e.g., the bound of ~\cite{anthony1993result}).
The technique herein avoids these complications; it is simpler than concentration-based approaches, practical to implement, and tight up to a factor of $\mathcal{O}(1/n)$, which is comparatively faster than concentration would allow.
Finally, herein we target distribution-free finite-sample control of~\eqref{eq:intro-risk-control}, but as a side-note it is also worth pointing the reader to the rich literature on functional central limit theorems~\citep{davidson2000functional}, which are another way of estimating risk functions.

\section{Theory}
\label{sec:theory}

In this section, we establish the core theoretical properties of conformal risk control. All proofs, unless otherwise specified, are deferred to Appendix~\ref{app:proofs}. 

\subsection{Risk control}
\label{sec:monotone}

We first show that the proposed algorithm leads to risk control when the loss is monotone. 
\begin{theorem}\label{thm:upper-bound}
  Assume that $L_{i}(\lambda)$ is non-increasing in $\lambda$, right-continuous, and
  \begin{equation}
    \label{eq:gF}
    L_{i}(\lambda_{\max}) \le \alpha, \quad \sup_{\lambda}L_{i}(\lambda) \le B < \infty \text{ almost surely}. 
  \end{equation}
  Then
  \begin{equation}
      \E[L_{n+1}(\lhat)] \leq \alpha.
  \end{equation}
\end{theorem}
\begin{proof}
  Let $\Rhat_{n+1}(\lambda) = (L_1(\lambda) + \ldots + L_{n+1}(\lambda))/ (n+1)$ and 
  \begin{equation}
    \lhat'= \inf\left\{ \lambda\in\Lambda : \Rhat_{n+1}(\lambda) \leq \alpha \right\}.
  \end{equation}
  Since $\inf_{\lambda} L_{i}(\lambda) = L_i(\lambda_{\max}) \leq \alpha$, $\lhat'$ is well-defined almost surely. 
  Since $L_{n+1}(\lambda) \leq B$, we know $\Rhat_{n+1}(\lambda) = \frac{n}{n+1}\Rhat_{n}(\lambda) + \frac{L_{n+1}(\lambda)}{n+1} \leq \frac{n}{n+1}\Rhat_{n}(\lambda) + \frac{B}{n+1}$.
  Thus,
  \begin{equation}
    \frac{n}{n+1} \Rhat_{n}(\lambda)  + \frac{B}{n+1} \leq \alpha \Longrightarrow \Rhat_{n+1}(\lambda) \leq \alpha.
  \end{equation}
  This implies $\lhat'\le \lhat$ when the LHS holds for some $\lambda\in \Lambda$. When the LHS is above $\alpha$ for all $\lambda\in \Lambda$, by definition, $\lhat = \lambda_{\max}\ge \lhat'$. Thus, $\lhat' \le \lhat$ almost surely. Since $L_{i}(\lambda)$ is non-increasing in $\lambda$,
  \begin{equation}
    \label{eq:lhat-lhat'}
    \E\Big[L_{n+1}\big(\lhat\big)\Big]\le \E\Big[L_{n+1}\big(\lhat'\big)\Big].
  \end{equation}
 Let $E$ be the multiset of loss functions $\{ L_1, \ldots, L_{n+1} \}$. Then $\lhat'$ is a function of $E$, or, equivalently, $\lhat'$ is a constant conditional on $E$. 
 Additionally, $L_{n+1}(\lambda)|E \sim \mathrm{Uniform}(\{L_1, ..., L_{n+1}\}$) by exchangeability.
 These facts combined with the right-continuity of $L_{i}$ imply
 \begin{equation}
    \E\left[L_{n+1}(\lhat')\mid E\right] = \frac{1}{n+1}\sum_{i=1}^{n+1}L_{i}(\lhat') \leq \alpha.
 \end{equation}
 The proof is completed by the law of total expectation and \eqref{eq:lhat-lhat'}.
\end{proof}

\subsection{A risk lower bound}
Next we show that the conformal risk control procedure is tight up to a factor $2B/(n+1)$.
Like the standard conformal coverage upper bound, the proof will rely on a form of continuity that prohibits large jumps in the risk function.
Towards that end, we will define the \emph{jump function} below, which quantifies the size of the discontinuity in a right-continuous input function $l$ at point $\lambda$:
\begin{equation}
    J(l,\lambda) = \underset{\epsilon \to 0^+}{\lim} l(\lambda - \epsilon) - l(\lambda)
\end{equation}
The jump function measures the size of a discontinuity at $l(\lambda)$.
When there is a discontinuity and $l$ is non-increasing, $J(l,\lambda) > 0$. 
When there is no discontinuity, the jump function is zero.
The next theorem will assume that the probability that $L_i$ has a discontinuity at any pre-specified $\lambda$ is $\mathbb{P}(J(L_i, \lambda) > 0) = 0$.
Under this assumption the conformal risk control procedure is not too conservative.
\begin{theorem}
  \label{thm:lower-bound}
  In the setting of Theorem~\ref{thm:upper-bound}, further assume that the $L_i$ are i.i.d., $L_i\geq 0$, 
  and for any $\lambda$, $\P\left(J(L_i, \lambda) > 0 \right) = 0$.
  Then 
  \begin{equation*}
    \E\Big[L_{n+1}\big(\lhat\big)\Big] \geq \alpha - \frac{2B}{n+1}.
  \end{equation*}
\end{theorem}

\subsection{Conformal prediction reduces to risk control}
\label{sec:cp_is_crc}
Conformal prediction can be thought of as controlling the expectation of an indicator loss function. 
Recall that the risk upper bound~\eqref{eq:risk-upper-bound} specializes to the conformal coverage guarantee in~\eqref{eq:miscoverage} when the loss function is the indicator of a miscoverage event. 
The conformal risk control procedure specializes to conformal prediction under this loss function as well. 
However, the risk lower bound in Theorem~\ref{thm:lower-bound} has a slightly worse constant than the usual conformal guarantee.
We now describe these correspondences.

First, we show the equivalence of the algorithms.
In conformal prediction, we have conformal scores $s(X_i,Y_i)$ for some score function $s : \mathcal{X} \times \mathcal{Y} \to \mathbb{R}$. 
Based on this score function, we create prediction sets for the test point $X_{n+1}$ as
\begin{equation*}
\Clhat(X_{n+1}) = \big\{y : s(X_{n+1},y) \leq \lhat\big\},
\end{equation*}
where $\lhat$ is the conformal quantile, a parameter that is set based on the calibration data. 
In particular, conformal prediction chooses $\lhat$ to be the $\lceil(n+1)(1-\alpha)\rceil/n$ sample quantile of $\{s(X_i,Y_i)\}_{i=1}^n$. 
To formulate this in the language of risk control, we consider a \emph{miscoverage loss} $L^{\rm Cvg}_i(\lambda) = \ind{Y_i \notin \Chat_\lambda(X_i)} = \ind{ s(X_i,Y_i) > \lambda }$.
Direct calculation of $\lhat$ from~\eqref{eq:lhat} then shows the equivalence of the proposed procedure to conformal prediction:
\begin{multline}
    \lhat = \inf\left\{\lambda : \frac{1}{n+1}\sum\limits_{i=1}^n\ind{ s(X_i,Y_i) > \lambda } + \frac{1}{n+1} \leq \alpha \right\} = \\ \underbrace{\inf\left\{\lambda : \frac{1}{n}\sum\limits_{i=1}^n\ind{ s(X_i,Y_i) \leq \lambda } \geq \frac{\lceil(n+1)(1-\alpha)\rceil}{n} \right\}}_{\rm conformal\ prediction\ algorithm}.
\end{multline}

Next, we discuss how the risk lower bound relates to its conformal prediction equivalent.
In the setting of conformal prediction,~\cite{lei2018distribution} proves that $\P( Y_{n+1} \notin \Clhat(X_{n+1})) \geq \alpha-1/(n+1)$ when the conformal score function follows a continuous distribution.
Theorem~\ref{thm:lower-bound} recovers this guarantee with a slightly worse constant: $\P( Y_{n+1} \notin \Clhat(X_{n+1})) \geq \alpha-2/(n+1)$.
First, note that our assumption in Theorem~\ref{thm:lower-bound} about the distribution of discontinuities specializes to the continuity of the score function when the miscoverage loss is used:\looseness=-1
\begin{equation}
    \mathbb{P}\left(J\Big(L^{\rm Cvg}_i, \lambda\Big) > 0\right) = 0 \Longleftrightarrow \mathbb{P}(s(X_i,Y_i) = \lambda) = 0.
\end{equation}
However, the bound for the conformal case is better than the bound for the general case in Theorem~\ref{thm:lower-bound} by a factor of two.
We do not know whether this factor of $2$ is improvable.
However, this factor is of little practical importance, as the difference between $1 / (n+1)$ and $2 / (n+1)$ is small even for moderate values of $n$.

\subsection{Controlling general loss functions}
\label{sec:counterexample}
We next show that the conformal risk control algorithm  does \emph{not} control the risk if the $L_i$ are not assumed to be monotone. 
In particular,~\eqref{eq:intro-risk-control} does not hold.
We show this by example.\looseness=-1
\begin{prop}
    \label{prop:counterexample}
    For any $\epsilon$, there exists a non-monotone loss function such that
    \begin{equation}
        \E\left[ L_{n+1}\big(\lhat\big) \right] \geq B-\epsilon.
    \end{equation}
\end{prop}

Notice that for any desired level $\alpha$, the expectation in~\eqref{eq:intro-risk-control} can be arbitrarily close to $B$. 
Since the function values here are in $[0,B]$, this means that even for bounded random variables, risk control can be violated by an arbitrary amount---unless further assumptions are placed on the $L_i$.
However, the algorithms developed may still be appropriate for near-monotone loss functions. 
Simply `monotonizing' all loss functions $L_i$ and running conformal risk control will guarantee~\eqref{eq:intro-risk-control}, but this strategy will only be powerful if the loss is near-monotone.
For concreteness, we describe this procedure below as a corollary of Theorem~\ref{thm:upper-bound}.\looseness=-1
\begin{cor}
    Allow $L_i(\lambda)$ to be any (possibly non-monotone) function of $\lambda$ satisfying~\ref{eq:gF}. 
    Take
    \begin{equation}
        \tilde{L}_i(\lambda) = \underset{\lambda' \geq \lambda}{\sup} L_i(\lambda'), \ \ \tilde{R}_n(\lambda) = \frac{1}{n}\sum\limits_{i=1}^n \tilde{L}_i(\lambda)
        \ \ \text{and } \tilde{\lambda} = \inf\left\{ \lambda : \frac{n}{n+1} \tilde{R}_n(\lambda) + \frac{B}{n+1} \leq \alpha \right\}.
    \end{equation}
    Then,
    \begin{equation}
        \E\left[ L_{n+1}\big( \tilde{\lambda} \big) \right] \leq \alpha.
    \end{equation}
\end{cor}
If the loss function is already monotone, then $\tilde{\lambda}$ reduces to $\lhat$.
We propose a further algorithm for picking $\lambda$ in Appendix~\ref{app:monotonized} that provides an asymptotic risk-control guarantee for \emph{non-monotone} loss functions.
However, this algorithm again is only powerful when the risk $\mathbb{E}[L_{n+1}(\lambda)]$ is near-monotone and reduces to the standard conformal risk control algorithm when the loss is monotone.

\section{Examples}
\label{sec:examples}
To demonstrate the flexibility and empirical effectiveness of the proposed algorithm, we apply it to four tasks across computer vision and natural language processing. 
All four loss functions are non-binary, monotone losses bounded by $1$.
They are commonly used within their respective application domains.
Our results validate that the procedure bounds the risk as desired and gives useful outputs to the end-user.
We note that the choices of $\Clam$ used herein are \emph{only for the purposes of illustration}; any nested family of sets will work. For each example use case,  for a representative $\alpha$ (details provided for each task) we provide both qualitative results, as well as quantitative histograms of the risk and set sizes over 1000 random data splits that demonstrate valid risk control (i.e., with mean $\leq \alpha$).
Code to reproduce our examples is available at the following GitHub link:~\url{https://github.com/aangelopoulos/conformal-risk}.

\subsection{FNR control in tumor segmentation}
\begin{figure}[!htb]
    \centering
    \includegraphics[width=\textwidth]{figures/multipolyp_grid_fig.pdf}
    \includegraphics[width=\textwidth]{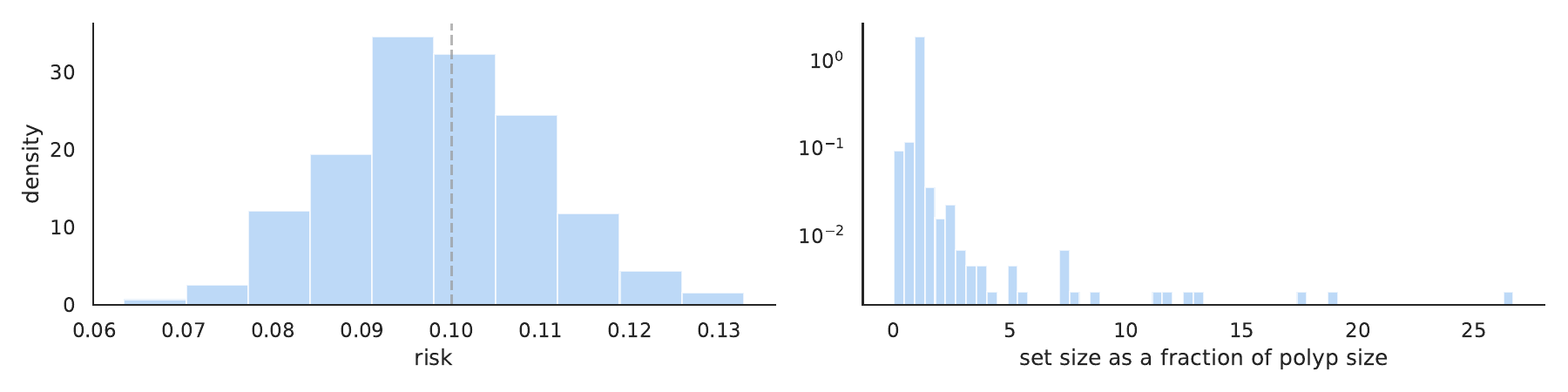}
    \caption{\textbf{FNR control in tumor segmentation}. The top figure shows examples of our procedure with correct pixels in white, false positives in blue, and false negatives in red. The bottom plots report FNR and set size over 1000 independent random data splits. The dashed gray line marks $\alpha$.}
    \label{fig:polyps}
\end{figure}

In the tumor segmentation setting, our input is a $d \times d$ image and our label is a set of pixels $Y_i \in \wp\left(\{(1,1), (1,2), ..., (d, d)\}\right)$, where $\wp$ denotes the power set.
We build on an image segmentation model $f : \X \to [0,1]^{d \times d}$ outputting a probability for each pixel and measure loss as the fraction of false negatives,
\begin{equation}
    \label{eq:fnp}
    L^{\mathrm{FNR}}_i(\lambda) = 1 - \frac{|Y_{i} \cap \Clam(X_{i})|}{ |Y_i|}, \text{ where } \Clam(X_{i}) = \left\{ y : f(X_{i})_y \geq 1-\lambda \right\}.
\end{equation}
The expected value of $L^{\mathrm{FNR}}_i$ is the FNR.
Since $L^{\mathrm{FNR}}_i$ is monotone, so is the FNR.
Thus, we use the technique in Section~\ref{sec:monotone} to pick $\lhat$ by~\eqref{eq:lhat} that controls the FNR on a new point, resulting in the following guarantee:
\begin{equation}
\label{eq:segmentation_FNR_control}
    \E\Big[L^{\mathrm{FNR}}_{n+1}(\lhat)\Big] \leq \alpha.
\end{equation}

For evaluating the proposed procedure we pool data from several online open-source gut polyp segmentation datasets: Kvasir, Hyper-Kvasir, CVC-ColonDB, CVC-ClinicDB, and ETIS-Larib. 
We choose a PraNet~\citep{fan2020pranet} as our base model $f$ and used $n=1000$, and evaluated risk control with the $781$ remaining validation data points.
We report results with $\alpha=0.1$ in Figure~\ref{fig:polyps}.
The mean and standard deviation of the risk over 1000 trials are 0.0987 and 0.0114, respectively.

\subsection{FNR control in multilabel classification}
\begin{figure}[ht]
    \centering
    \includegraphics[width=\textwidth]{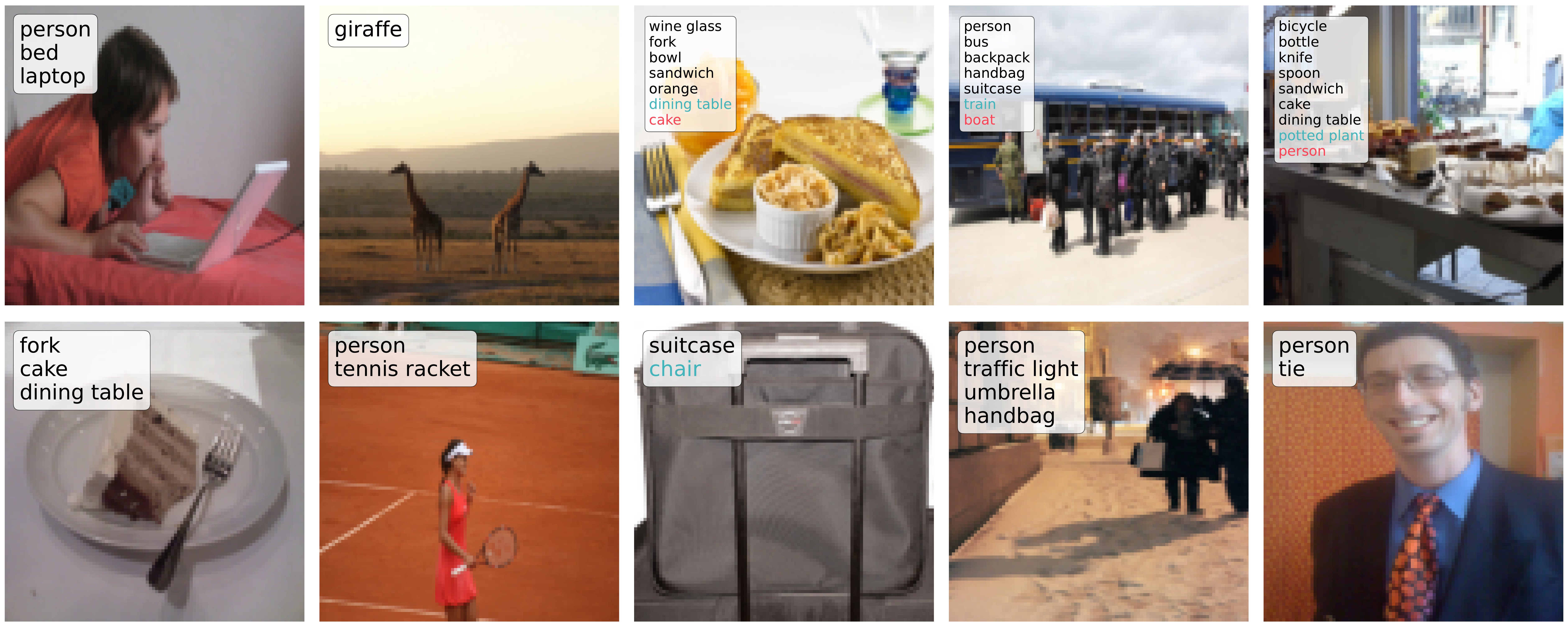}
    \includegraphics[width=\textwidth]{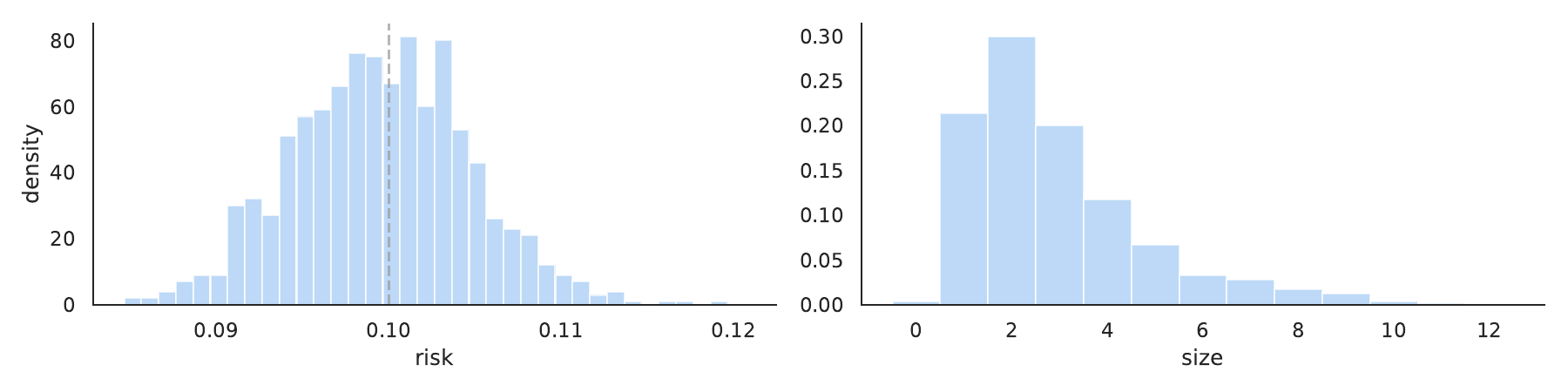}
    \caption{\textbf{FNR control on MS COCO}. The top figure shows examples of our procedure with correct classes in black, false positives in blue, and false negatives in red. The bottom plots report FNR and set size over 1000 independent random data splits. The dashed gray line marks $\alpha$.}
    \label{fig:coco}
\end{figure}
In the multilabel classification setting, our input $X_i$ is an image and our label is a set of classes $Y_i \subset \{1,\dots,K\}$ for some number of classes $K$.
Using a multiclass classification model $f : \X \to [0,1]^K$, we form prediction sets and calculate the number of false positives exactly as in~\eqref{eq:fnp}.
By Theorem~\ref{thm:upper-bound}, picking $\lhat$ as in~\eqref{eq:lhat} again yields the FNR-control guarantee in~\eqref{eq:segmentation_FNR_control}.

We use the Microsoft Common Objects in Context (MS COCO) computer vision dataset~\citep{lin2014microsoft}, a large-scale 80-class multiclass classification baseline dataset commonly used in computer vision, to evaluate the
proposed procedure. 
We choose a TResNet~\citep{ridnik2020tresnet} as our base model $f$ and used $n=4000$, and evaluated risk control with 1000 validation data points.
We report results with $\alpha=0.1$ in Figure~\ref{fig:coco}.
The mean and standard deviation of the risk over 1000 trials are 0.0996 and 0.0052, respectively.

\subsection{Control of graph distance in hierarchical image classification}
\begin{figure}[ht]
    \centering
    \includegraphics[width=\textwidth]{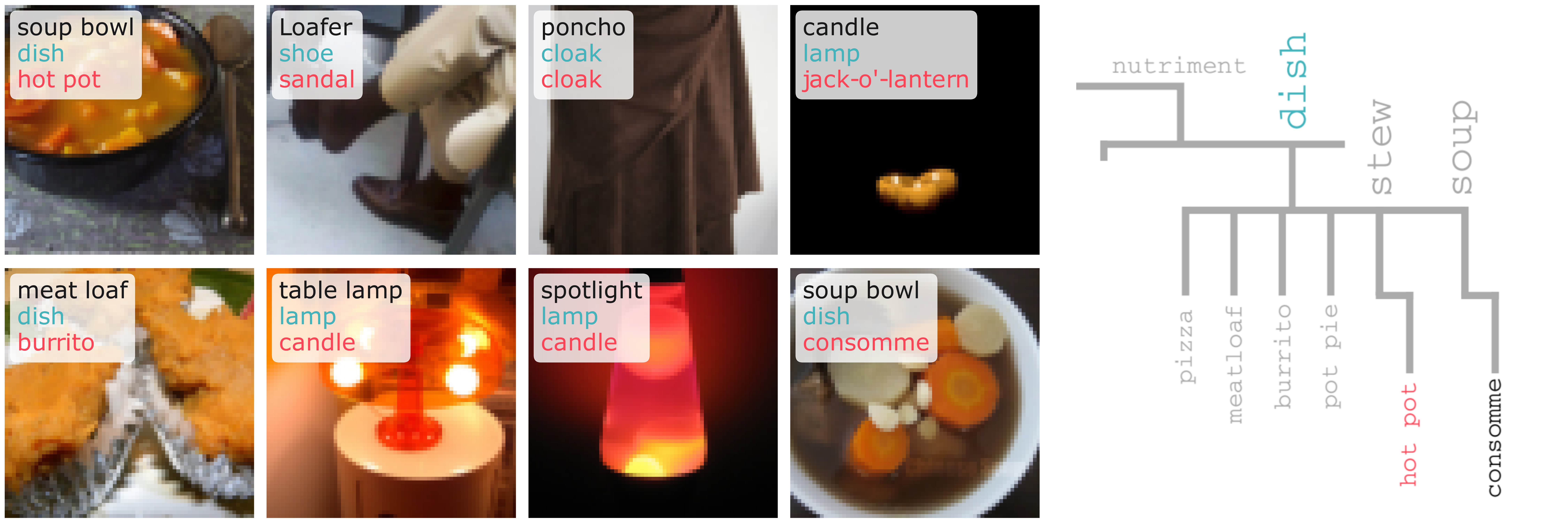}
    \includegraphics[width=\textwidth]{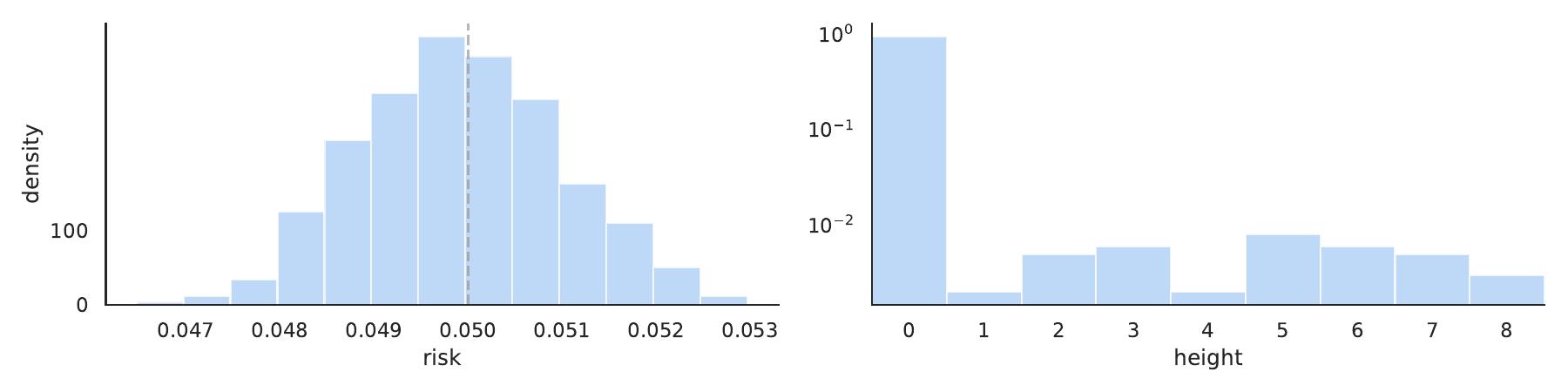}
    \caption{\textbf{Control of graph distance on hierarchical ImageNet}. The top figure shows examples of our procedure with correct classes in black, false positives in blue, and false negatives in red. The bottom plots report our minimum hierarchical distance loss and set size over 1000 independent random data splits. The dashed gray line marks $\alpha$.}
    \label{fig:hierarchical}
\end{figure}
In the $K$-class hierarchical classification setting, our input $X_i$ is an image and our label is a leaf node $Y_i \in \{1, ..., K\}$ on a tree with nodes $\mathcal{V}$ and edges $\mathcal{E}$.
Using a single-class classification model $f : \X \to \Delta^K$, we calibrate a loss in graph distance between the interior node we select and the closest ancestor of the true class.
For any $x \in \X$, let $\hat{y}(x) = \arg\max_{k} f(x)_k$ be the class with the highest estimated probability. 
Further, let $d:\mathcal{V} \times \mathcal{V} \to \mathbb{Z}$ be the function that returns the length of the shortest path between two nodes, let $\mathcal{A}: \mathcal{V} \to 2^\mathcal{V}$ be the function that returns the ancestors of its argument, and let $\mathcal{P}: \mathcal{V} \to 2^\mathcal{V}$ be the function that returns the set of leaf nodes that are descendants of its argument. 
We also let $g(v,x) = \underset{ k \in \mathcal{P}(v) }{\sum}f(x)_k$ be the sum of scores of leaves descended from $v$.
Further, define a hierarchical distance 
\begin{equation*}
d_H(v,u) = \underset{a \in \mathcal{A}(v)}{\inf} \{ d(a,u) \}.
\end{equation*}
For a set of nodes $\Clam \in 2^\mathcal{V}$, we then define the set-valued loss 
\begin{equation}
    \label{eq:hierarchical-loss}
    L_i^{\rm Graph}(\lambda) = \underset{s \in \Clam(X_i)}{\inf} \{ d_H(y,s) \} / D,\text{ where } \Clam(x) = \underset{\{a \in \mathcal{A}(\hat{y}(x)) \ : \ g(a,x) \geq -\lambda\}}{\bigcap}  \mathcal{P}(a).
\end{equation}
This loss returns zero if $y$ is a child of any element in $\Clam$, and otherwise returns the minimum distance between any element of $\Clam$ and any ancestor of $y$, scaled by the depth $D$.
Thus, it is a monotone loss function and can be controlled by choosing $\lhat$ as in~\eqref{eq:lhat} to achieve the guarantee
\begin{equation}
    \E\Big[L^{\mathrm{Graph}}_{n+1}(\lhat)\Big] \leq \alpha.
\end{equation}

For this experiment, we use the ImageNet dataset~\citep{deng2009imagenet}, which comes with an existing label hierarchy, WordNet, of maximum depth $D=14$. 
We choose a ResNet152~\citep{he2016deep} as our base model $f$ and used $n=30000$, and evaluated risk control with the remaining $20000$.
We report results with $\alpha=0.05$ in Figure~\ref{fig:hierarchical}.
The mean and standard deviation of the risk over 1000 trials are 0.0499 and 0.0011, respectively.

\subsection{F1-score control in open-domain question answering}
\begin{figure}[ht]
    \centering
    \includegraphics[width=\textwidth]{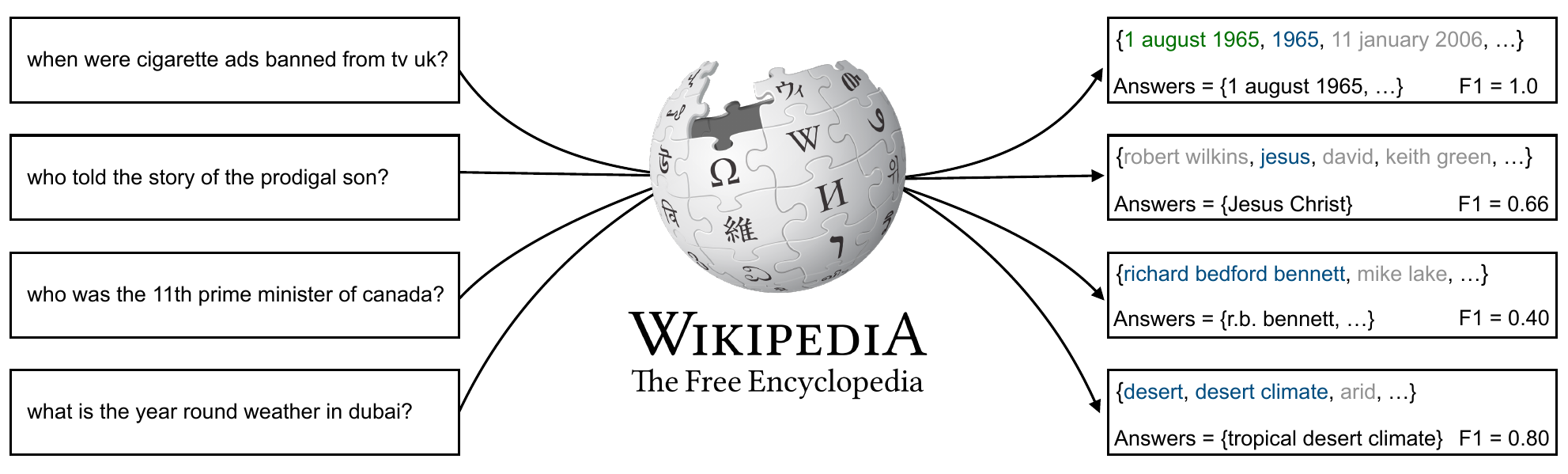}
    \includegraphics[width=\textwidth]{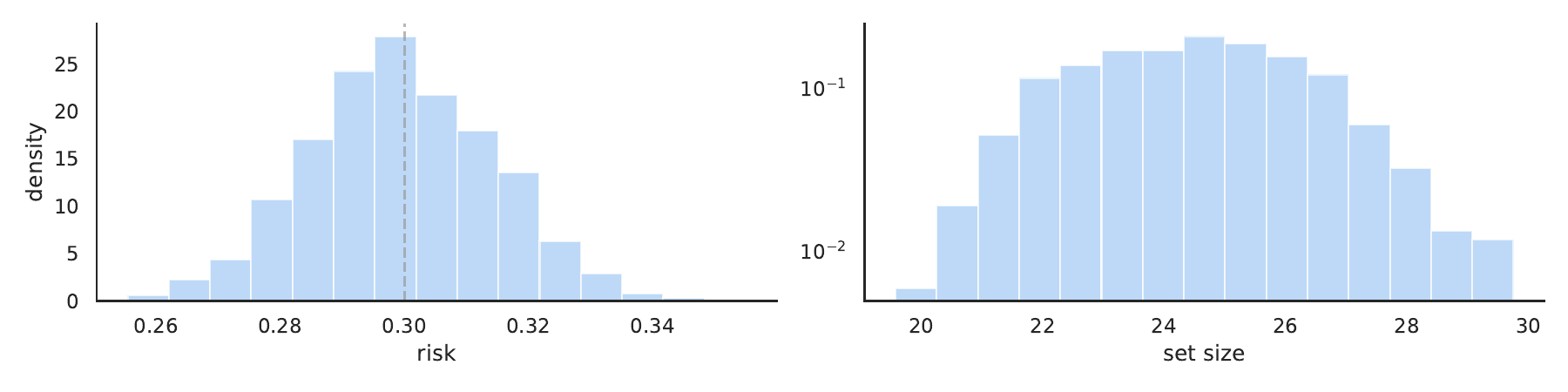}
    \caption{\textbf{F1-score control on Natural Questions}. The top figure shows examples of our procedure with fully correct answers in green, partially correct answers in blue, and false positives in gray. Note that due to the nature of the evaluation, answers that are technically correct may still be down-graded if they do not match the reference. We treat this as part of the randomness in the task. The bottom plots report the F1 risk and average set size over 1000 independent random data splits. The dashed gray line marks $\alpha$.}
    \label{fig:qa}
\end{figure}

In the open-domain question answering setting, our input $X_i$ is a question and our label $Y_i$ is a set of (possibly non-unique) correct answers. For example, the input $$X_{n+1} = \text{``Where was Barack Obama Born?''}$$ could have the answer set  $$Y_{n+1} = \{\text{``Hawaii'', ``Honolulu, Hawaii'', ``Kapo'olani Medical Center''}\}$$ Formally, here we treat all questions and answers as being composed of sequences (up to size $m$) of tokens in a vocabulary $\mathcal{V}$---i.e., assuming $k$ valid answers, we have $X_i \in \mathcal{Z}$ and $Y_i \in \mathcal{Z}^k$, where $\mathcal{Z}:= \mathcal{V}^m$. Using an open-domain question answering model that individually scores candidate output answers $f \colon \mathcal{Z} \times \mathcal{Z} \rightarrow \mathbb{R}$, we calibrate the \emph{best} token-based F1-score of the prediction set, taken over all pairs of predictions and answers:
\begin{equation}
\begin{split}
    L_i^{\mathrm{F1}}(\lambda) &= 1 - \max\big\{ \mathrm{F1}(a, c) \colon c \in \Clam(X_i), a \in Y_i\big\}, \\&\text{ where } \Clam(X_{i}) = \left\{ y \in \mathcal{V}^m: f(X_{i}, y) \geq \lambda \right\}.
\end{split}
\end{equation}
We define the F1-score following popular QA evaluation metrics~\citep{squad}, where we treat predictions and ground truth answers as bags of tokens and compute the geometric average of their precision and recall (while ignoring punctuation and articles $\{\text{``a'', ``an'', ``the''}\}$). Since $L_i^\mathrm{F1}$, as defined in this way, is monotone and upper bounded by $1$, it can  be controlled by choosing $\hat{\lambda}$ as in Section~\ref{sec:monotone} to achieve the following guarantee:\looseness=-1
\begin{equation}
    \mathbb{E}\left[L_{n+1}^{\mathrm{F1}}(\hat{\lambda})\right] \leq \alpha.
\end{equation}

We use the Natural Questions (NQ) dataset~\citep{nq}, a popular open-domain question answering baseline, to evaluate our method. We use the splits  distributed as part of the Dense Passage Retrieval (DPR) package~\citep{dpr}. Our base model is the DPR Retriever-Reader model~\citep{dpr}, which retrieves passages from Wikipedia that might contain the answer to the given query, and then uses a reader model to extract text sub-spans from the retrieved passages that serve as candidate answers. Instead of enumerating all possible answers to a given question (which is intractable), we retrieve the top several hundred candidate answers, extracted from the top 100 passages (which is sufficient to control all risks of interest). We use $n = 2500$ calibration points, and evaluate risk control with the remaining $1110$. We report results with $\alpha=0.3$ (chosen empirically as the lowest F1 score which typically results in nearly correct answers) in Figure~\ref{fig:qa}. The mean and standard deviation of the risk over 1000 trials are 0.2996 and 0.0150, respectively.\looseness=-1

\section{Extensions}
\label{sec:extensions}
In this section, we discuss several theoretical extensions of our procedure.

\subsection{Risk control under distributional shift}

Suppose the researcher wants to control the risk under a distribution shift.
Then the goal in~\eqref{eq:intro-risk-control} can be redefined as 
\begin{equation}\label{eq:weighted_objective}
   \E_{(X_1,Y_1), \ldots, (X_n, Y_n) \sim P_{\mathrm{train}}, \; (X_{n+1}, Y_{n+1})\sim P_{\mathrm{test}}}\Big[L_{n+1}\big(\lhat \big)\Big] \leq \alpha, 
\end{equation}
where $P_{\mathrm{test}}$ denotes the test distribution that is different from the training distribution $P_{\mathrm{train}}$ that $(X_i, Y_i)_{i=1}^{n}$ are sampled from. Assuming that $P_{\mathrm{test}}$ is absolutely continuous with respect to $P_{\mathrm{train}}$, the weighted objective \eqref{eq:weighted_objective} can be rewritten as 
\begin{equation}\label{eq:weighted_objective_equiv}
\begin{aligned}
\E_{(X_{1}, Y_{1}), \ldots, (X_{n+1}, Y_{n+1})\sim P_{\mathrm{train}}}\Big[w(X_{n+1}, Y_{n+1})L_{n+1}\big(\lhat \big)\Big] &\leq \alpha, \\ \text{where } w(x, y) &= \frac{dP_{\mathrm{test}}(x, y)}{dP_{\mathrm{train}}(x, y)}.
\end{aligned}
\end{equation}
When $w$ is known and bounded, we can apply our procedure on the loss function $\tilde{L}_{n+1}(\lambda) = w(X_{n+1}, Y_{n+1})L_{n+1}(\lambda)$, which is non-decreasing, bounded, and right-continuous in $\lambda$ whenever $L_{n+1}$ is. 
Thus, Theorem \ref{thm:upper-bound} guarantees that the resulting $\hat{\lambda}$ satisfies \eqref{eq:weighted_objective_equiv}.

In the setting of transductive learning, $X_{n+1}$ is available to the user. If the conditional distribution of $Y$ given $X$ remains the same in the training and test domains, the distributional shift reduces to a covariate shift and 
\[w(X_{n+1}, Y_{n+1}) = w(X_{n+1}) \triangleq \frac{dP_{\mathrm{test}}(X_{n+1})}{dP_{\mathrm{train}}(X_{n+1})}.\]
In this case, we can achieve the risk control even when $w$ is unbounded. In particular, assuming $L_i\in [0, B]$, for any potential value $x$ of the covariate, we define
\begin{equation}\label{eq:weighted_lhat}
\lhat(x) = \inf\left\{ \lambda : \frac{\sum_{i=1}^{n}w(X_i)L_i(\lambda) + w(x)B}{\sum_{i=1}^{n}w(X_i) + w(x)} \leq \alpha \right\}.
\end{equation}
When $\lambda$ does not exist, we simply set $\lhat(x) = \max\Lambda$.
It is not hard to see that $\lhat(x)\equiv \lhat$ in the absence of covariate shifts.
We can prove the following result.
\begin{prop}\label{thm:upper-bound-weighted}
  In the setting of Theorem \ref{thm:upper-bound},
  \begin{equation}
      \E_{(X_1,Y_1), \ldots, (X_n, Y_n) \sim P_{\mathrm{train}}, (X_{n+1}, Y_{n+1})\sim P_{\mathrm{test}}}[L_{n+1}(\lhat(X_{n+1}))] \leq \alpha.
  \end{equation}
\end{prop}
It is easy to show that the weighted conformal procedure \citep{tibshirani2019conformal} is a special case with $L_i(\lambda) = \ind{Y_i \not \in \mathcal{C}_\lambda(X_i)}$ where $\mathcal{C}_\lambda(X_i)$ is the prediction set that thresholds the conformity score at $\lambda$. Thus, Proposition \ref{thm:upper-bound-weighted} generalizes \cite{tibshirani2019conformal} to any monotone risk. 
When the covariate shift $w(x)$ is unknown but unlabeled data in the test domain are available, it can be estimated, up to a multiplicative factor that does not affect $\lhat(x)$, by any probabilistic classification algorithm; see \cite{lei2020conformal} and \cite{candes2023conformalized} in the context of missing and censored data, respectively. We leave the full investigation of weighted conformal risk control with an estimated covariate shift for future research.\looseness=-1

\subsubsection*{Total variation bound}
Finally, for arbitrary distribution shifts, we give a total variation bound describing the way standard (unweighted) conformal risk control degrades.
The bound is analogous to that of~\cite{barber2022conformal}for independent but non-identically distributed data (see their Section 4.1), though the proof is different.
Here we will use the notation $Z_i = (X_i, Y_i)$, and $\lhat(Z_1, \ldots, Z_n)$ to refer to that chosen in~\eqref{eq:lhat}.

\begin{prop}
  \label{thm:tv-bound}
  Let $Z = (Z_1, \ldots, Z_{n+1})$ be a sequence of random variables.
  Then, under the conditions in Theorem \ref{thm:upper-bound},
 \begin{equation*}
      \E\left[ L_{n+1}(\lhat) \right] \leq \alpha + B\sum_{i=1}^{n}\mathrm{TV}(Z_i, Z_{n+1}).
  \end{equation*}
  If further the assumptions of Theorem \ref{thm:lower-bound} hold, 
    \begin{equation*}
      \E\left[ L_{n+1}(\lhat) \right] \geq \alpha - B\left( \frac{2}{n+1} + \sum_{i=1}^{n}\mathrm{TV}(Z_i, Z_{n+1}) \right).
  \end{equation*}
\end{prop}

\subsection{Quantile risk control}
\cite{snell2022quantile} generalizes \cite{bates2021distribution} to control the quantile of a monotone loss function conditional on $(X_i, Y_i)_{i=1}^{n}$ with probability $1 - \delta$ over the calibration dataset for any user-specified tolerance parameter $\delta$. In some applications, it may be sufficient to control the unconditional quantile of the loss function, which alleviates the burden of the user to choose the tolerance parameter $\delta$. 

For any random variable $X$, let \[\mathrm{Quantile}_{\beta}(X) = \inf\{x: \P(X \le x)\ge \beta\}.\]
Analogous to \eqref{eq:intro-risk-control}, we want to find $\lhat$ based on $(X_i, Y_i)_{i=1}^{n}$ such that
\begin{equation}\label{eq:quantile_risk}
\mathrm{Quantile}_{\beta}\lb L_{n+1}(\lhat_\beta)\rb\le \alpha.
\end{equation}
By definition,
\[\mathrm{Quantile}_{\beta}\lb L_{n+1}( \lhat_\beta)\rb\le \alpha \Longleftrightarrow \E\left[\ind {L_{n+1}(\lhat_\beta) > \alpha}\right]\le 1 - \beta.\]
As a consequence, quantile risk control is equivalent to expected risk control \eqref{eq:intro-risk-control} with loss function $\tilde{L}_i(\lambda) = \ind{L_i(\lambda) > \alpha}$. Let 
\[\lhat_\beta = \inf\left\{\lambda\in \Lambda: \frac{1}{n+1}\sum_{i=1}^{n}\ind{L_i(\lambda) > \alpha} + \frac{1}{n+1}\le 1 - \beta\right\}.\]
\begin{prop}\label{thm:quantile-risk-control}
  In the setting of Theorem \ref{thm:upper-bound}, \eqref{eq:quantile_risk} is achieved.
  \end{prop}

\cite{snell2022quantile} considers the high-probability control of a wider class of quantile-based risks which include the conditional value-at-risk (CVaR). It is unclear whether those more general risks can be controlled unconditionally. We leave this open problem for future research.
  
\subsection{Controlling multiple risks}
\label{sec:multiple-risks}
Let $L_{i}(\lambda; \gamma)$ be a family of loss functions indexed by $\gamma\in \Gamma$ for some domain $\Gamma$ that may have infinitely many elements. 
A researcher may want to control $\E[L_i(\lambda; \gamma)]$ at level $\alpha(\gamma)$.
Equivalently, we need to find an $\lhat$ based on $(X_i, Y_i)_{i=1}^{n}$ such that
\begin{equation}
    \label{eq:multiple-risk-goal}
    \sup_{\gamma\in \Gamma}\E\left[\frac{L_i(\lhat; \gamma)}{\alpha(\gamma)}\right]\le 1.
\end{equation}

Though the above worst-case risk is not an expectation, it can still be controlled.
Towards this end, we define
\begin{equation}
    \label{eq:lhat-multiple-risks}
    \lhat = \sup_{\gamma \in \Gamma} \lhat_{\gamma}, \text{ where }
    \lhat_{\gamma} = \inf\left\{\lambda : \frac{1}{n+1}\sum_{i=1}^{n}L_i(\lambda; \gamma) + \frac{B}{n+1}\le \alpha(\gamma) \right\}.
\end{equation}
Then the risk is controlled.
\begin{prop}
  \label{thm:multiple-risks}
  In the setting of Theorem~\ref{thm:upper-bound},~\eqref{eq:multiple-risk-goal} is satisfied.
\end{prop}

\subsection{Adversarial risks}
We next show how to control risks defined by adversarial perturbations.
We adopt the same notation as Section~\ref{sec:multiple-risks}.
\cite{bates2021distribution} (Section 6.3) discusses the adversarial risk where $\Gamma$ parametrizes a class of perturbations of $X_{n+1}$, e.g., $L_i(\lambda; \gamma) = L(X_i + \gamma, Y_i)$ and $\Gamma = \{\gamma: \|\gamma\|_{\infty}\le \eps\}$.
A researcher may want to find an $\lhat$ based on $(X_i, Y_i)_{i=1}^{n}$ such that
\begin{equation}
    \label{eq:adversarial-goal}
    \E[ \sup_{\gamma \in \Gamma} L_i(\lambda; \gamma)] \leq \alpha.
\end{equation}

This can be recast as a conformal risk control problem by taking $\tilde{L}_i(\lambda) = \sup_{\gamma \in \Gamma} L_i(\lambda; \gamma)$.
Then, the following choice of $\lambda$ leads to risk control:
\begin{equation}
    \label{eq:lhat-adversarial}
    \lhat = \inf\left\{\lambda : \frac{1}{n+1}\sum_{i=1}^{n}\tilde{L}_i(\lambda) + \frac{B}{n+1} \le \alpha \right\}.
\end{equation}
\begin{prop}
  \label{thm:adversarial}
  In the setting of Theorem~\ref{thm:upper-bound},~\eqref{eq:adversarial-goal} is satisfied.
\end{prop}

\subsection{U-risk control}
For ranking and metric learning, \cite{bates2021distribution} considered loss functions that depend on two test points. In general, for any $k > 1$ and subset $\S\subset \{1, \ldots, n+k\}$ with $|\S| = k$, let $L_\S (\lambda)$ be a loss function. Our goal is to find $\lhat_k$ based on $(X_i, Y_i)_{i=1}^{n}$ such that 
\begin{equation}\label{eq:U-risk}
\E\left[L_{\{n+1, \ldots, n+k\}}(\lhat_k)\right]\le \alpha.    
\end{equation}
We call the LHS a U-risk since, for any fixed $\lhat_k$, it is the expectation of an order-$k$ U-statistic. As a natural extension, we can define 
\begin{equation}\label{eq:lhatk}
\lhat_k = \inf\left\{\lambda: \frac{k!n!}{(n+k)!}\sum_{\S\subset \{1, \ldots, n\}, |\S| = k}L_{\S}(\lambda) + B \lb 1 - \frac{(n!)^2}{(n+k)!(n-k)!}\rb\le \alpha\right\}.    
\end{equation}
Again, we define $\lhat_k = \lambda_{\max}$ when the right-hand side is an empty set. Then we can prove the following result. 
\begin{prop}\label{thm:U_risk_control}
  Assume that $L_{\S}(\lambda)$ is non-increasing in $\lambda$, right-continuous, and
  \begin{equation*}
    L_{\S}(\lambda_{\max}) \le \alpha, \quad \sup_{\lambda}L_{\S}(\lambda) \le B < \infty \text{ almost surely}. 
  \end{equation*}
  Then \eqref{eq:U-risk} is achieved.
\end{prop}

\section{Conclusion}
This generalization of conformal prediction broadens its scope to new applications, as shown in Section~\ref{sec:examples}.
The mathematical tools developed in Section~\ref{sec:theory}, Section~\ref{sec:extensions}, and the Appendix may be of independent technical interest, since they provide a new and more general language for studying conformal prediction along with new results about its validity. 

\section*{Acknowledgements}
The authors thank Christopher Yeh for pointing out that the factor of 2 in Theorem~\ref{thm:lower-bound} may not be tight, correcting an error in an earlier draft. 
The incorrect proposition stating that it was tight has been removed; the tightness of this factor remains an open question.
The authors would like to thank Nicolas Christianson, Amit Kohli, Sherrie Wang, and Tijana Zrni\'c for comments on early drafts.
A.~A.~would like to thank Ziheng (Tony) Wang for helpful conversations.
A.~A.~is funded by the NSF GRFP and a Berkeley Fellowship.
S.~B.~is supported by the NSF FODSI fellowship and the Simons institute.
A.~F. is partially funded by the NSF GRFP and MIT MLPDS.\looseness=-1

\printbibliography

\appendix

\section{Monotonizing non-monotone risks}
\label{app:monotonized}
We next show that the proposed algorithm leads to asymptotic risk control for non-monotone risk functions when applied to a monotonized version of the empirical risk.
We set the \emph{monotonized empirical risk} to be
\begin{equation}
  \Rhatplus_{n}(\lambda) = \underset{t \geq \lambda}{\sup}\;\Rhat_{n}(t),
\end{equation}
then define
\begin{equation}
    \label{eq:lhatplus}
    \lhatplus_n = \inf\left\{\lambda:  \Rhatplus_{n}(\lambda) \le \alpha \right\}.
\end{equation}

\begin{theorem}
  \label{thm:monotonized}
  Let the $L_{i}(\lambda)$ be right-continuous, i.i.d., bounded (both above and below) functions satisfying~\eqref{eq:gF}.
  Then,
  \begin{equation}
    \underset{n \to \infty}{\lim}\E\Big[L_{n+1}\big(\lhatplus_n\big)\Big]\le \alpha.
  \end{equation}
\end{theorem}

Theorem~\ref{thm:monotonized} implies that an analogous procedure to~\ref{eq:lhat} also controls the risk asymptotically.
In particular, taking
\begin{equation}
    \tilde{\lambda}^{\uparrow} = \inf\left\{\lambda:  \Rhatplus_{n}(\lambda) + \frac{B}{n+1} \le \alpha \right\}
\end{equation}
also results in asymptotic risk control (to see this, plug $\tilde{\lambda}^{\uparrow}$ into Theorem~\ref{thm:monotonized} and see that the  risk level is bounded above by $\alpha-\frac{B}{n+1}$).
Note that in the case of a monotone loss function, $\tilde{\lambda}^{\uparrow} = \lhat$.
However, the counterexample in Proposition~\ref{prop:counterexample} does not apply to $\tilde{\lambda}^{\uparrow}$, and it is currently unknown whether this procedure does or does not provide finite-sample risk control.

\section{Proofs}
\label{app:proofs}

The proof of Theorem~\ref{thm:lower-bound} uses the following lemma on the approximate continuity of the empirical risk.
\begin{lemma}[Jump Lemma]
  \label{lem:jump}
  In the setting of Theorem~\ref{thm:lower-bound}, any jumps in the empirical risk are bounded, i.e.,
  \begin{equation}
     \sup_{\lambda}J\big(\Rhat_{n}, \lambda\big) \overset{a.s.}{\leq} \frac{B}{n}.
  \end{equation}
\end{lemma}

\begin{proof}[Proof of Jump Lemma, Lemma~\ref{lem:jump}]
  By boundedness, the maximum contribution of any single point to the jump is $\frac{B}{n}$, so
  \begin{equation}
    \exists \lambda :\; J\big(\Rhat_n, \lambda \big) > \frac{B}{n} \\ \Longrightarrow \exists \lambda :\; J(L_i,\lambda) > 0 \text{ and } J(L_j,\lambda) > 0 \text{ for some } i \neq j.
  \end{equation}
      Call $\D_i = \{ \lambda : J(L_i, \lambda) > 0 \}$ the sets of discontinuities in  $L_i$. Since $L_i$ is bounded monotone, $\D_i$ has countably many points. The union bound then implies that
    \begin{equation}
        \P\left(\exists \lambda : \;  J(\Rhat_n, \lambda) > \frac{B}{n} \right) \le \sum_{i\neq j}\P(\D_i \cap \D_j \neq \emptyset)
    \end{equation}
  Rewriting each term of the right-hand side using tower property and law of total probability gives
  \begin{align}
      \P\left( \D_i \cap \D_j \neq \emptyset \right) 
      &= \E\left[ \P\big( \D_i \cap \D_j \neq \emptyset \: \big| \: \D_j \big) \right] \\ 
      &\leq \E\left[ \sum\limits_{\lambda \in \D_j} \P\left( \lambda \in \D_i \; \Big| \; \D_j \right) \right] = \E\left[ \sum\limits_{\lambda \in \D_j} \P\left( \lambda \in \D_i\right) \right],
  \end{align}
  Where the second inequality is because the union of the events $\lambda \in \D_j$ is the entire sample space, but they are not disjoint, and the third equality is due to the independence between $\D_i$ and $\D_j$. Rewriting in terms of the jump function and applying the assumption $\P \left( J(L_i, \lambda) > 0 \right) = 0$,
  \begin{equation}
      \E\left[ \sum\limits_{\lambda \in \D_j} 
      \P \left( \lambda \in \D_i\right) \right] = \E\left[ \sum\limits_{\lambda \in \D_j} 
      \P \left( J(L_i, \lambda) > 0 \right) \right] = 0.
  \end{equation}
  Chaining the above inequalities yields $\P\left(\exists \lambda :   J(\Rhat_n, \lambda) > \frac{B}{n} \right) \leq 0$, so \\ $\P\left(\exists \lambda :   J(\Rhat_n, \lambda) > \frac{B}{n} \right) = 0$.
\end{proof}

\begin{proof}[Proof of Theorem~\ref{thm:lower-bound}]
  If $L_i(\lambda_{\max})\ge \alpha - 2B/(n+1)$, then $\E[L_{n+1}(\hat{\lambda})]\ge \alpha - 2B/(n+1)$. Throughout the rest of the proof, we assume that $L_i(\lambda_{\max}) < \alpha - 2B/(n+1)$.
  Define the quantity
  \begin{equation}
    \lhat'' = \inf\left\{ \lambda : \Rhat_{n+1}(\lambda) + \frac{B}{n+1} \leq \alpha \right\}.
  \end{equation}
  Since $L_i(\lambda_{\max}) < \alpha - 2B/(n+1) < \alpha - B/ (n+1)$, $\lhat''$ exists almost surely. Deterministically, $\frac{n}{n+1} \Rhat_n(\lambda) \leq \Rhat_{n+1}(\lambda)$, which yields $\lhat \leq \lhat''$.
  Again since $L_i(\lambda)$ is non-increasing in $\lambda$,
  \begin{equation}
    \E\left[ L_{n+1}\big(\lhat''\big) \right] \leq \E\left[ L_{n+1}\big(\lhat\big) \right]
  \end{equation}
  By exchangeability and the fact that $\lhat''$ is a symmetric function of $L_1, \ldots, L_{n+1}$,
  \begin{equation}
    \E\left[ L_{n+1}\big(\lhat''\big) \right] = \E\left[ \Rhat_{n+1}\big(\lhat''\big) \right]
  \end{equation}
   
  For the remainder of the proof we focus on lower-bounding $\Rhat_{n+1}\big(\lhat''\big)$. 
  We begin with the following identity: 
  \begin{equation}
    \alpha = \Rhat_{n+1}\big(\lhat''\big) + \frac{B}{n+1} - \Big(\Rhat_{n+1}\big(\lhat''\big) + \frac{B}{n+1} -\alpha\Big).
  \end{equation}
  Rearranging the identity,
  \begin{equation}
    \Rhat_{n+1}\big(\lhat''\big) = \alpha - \frac{B}{n+1} + \Big(\Rhat_{n+1}\big(\lhat''\big) + \frac{B}{n+1} -\alpha\Big).
  \end{equation}
  Using the Jump Lemma to bound $\Big(\Rhat_{n+1}\big(\lhat''\big) + \frac{B}{n+1} -\alpha\Big)$ below by $-\frac{B}{n+1}$ gives 
  \begin{equation}
    \Rhat_{n+1}\big(\lhat''\big) \geq \alpha - \frac{2B}{n+1}.
  \end{equation}
  Finally, chaining together the above inequalities, 
  \begin{equation}
    \E\bigg[ L_{n+1}(\lhat) \bigg] \geq \E\bigg[ \Rhat_{n+1}(\lhat'') \bigg] \geq \alpha - \frac{2B}{n+1}.
  \end{equation}

\end{proof}

\begin{proof}[Proof of Proposition~\ref{prop:counterexample}]
  Without loss of generality, we assume $B=1$.
  Assume $\lhat$ takes values in $[0,1]$ and $\alpha \in (1/(n+1), 1)$. Let $p\in (0, 1)$, $N$ be any positive integer, and $L_{i}(\lambda)$ be i.i.d. right-continuous piecewise constant (random) functions with
\begin{equation}
  L_{i}(N/N) = 0, \quad \left(L_{i}(0/N), L_{i}(1/N), \ldots, L_{i}((N - 1)/N)\right)\stackrel{i.i.d.}{\sim}\text{Ber}(p).
\end{equation}
By definition, $\lhat$ is independent of $L_{n+1}$. Thus, for any $j = 0, 1, \ldots, N-1$,
\begin{equation}
  \left\{L_{n+1}(\lhat)\mid \lhat = j/N\right\} \sim \text{Ber}(p), \quad \left\{L_{n+1}(\lhat)\mid \lhat = 1\right\} \sim \delta_{0}.
\end{equation}
Then,
\begin{equation}
  \E\Big[ L_{n+1}\big(\lhat\big) \Big] = p\cdot\P(\lhat \neq 1)\\
\end{equation}
Note that 
\begin{equation}
  \lhat \neq 1 \Longleftrightarrow \min_{j \in \{0, \ldots, N-1\}}\frac{1}{n+1}\sum_{i=1}^{n}L_{i}(j / N) \le \alpha - \frac{1}{n+1}.
\end{equation}
Since $\alpha > 1 / (n + 1)$,
\begin{align*}
  \P(\lhat \neq 1) = 1 - \P(\lhat = 1) &= 1-\P\left(\text{for all }j, \text{ we have }\frac{1}{n+1}\sum_{i=1}^{n}L_{i}(j / N) > \alpha - \frac{1}{n+1} \right)\\
  & = 1 - \left(\sum\limits_{k=\lceil(n+1)\alpha\rceil}^n {n\choose k} p^k(1-p)^{(n-k)}\right)^N\\
  & = 1 - \left(1-\mathrm{BinoCDF}\big(n,p,\lceil(n+1)\alpha\rceil-1\big)\right)^N\\
\end{align*}
As a result, 
\begin{equation}
  \E\Big[ L_{n+1}\big(\lhat\big) \Big] = p \Bigg(1 - \left(1-\mathrm{BinoCDF}\big(n,p,\lceil(n+1)\alpha\rceil-1\big)\right)^N\Bigg).
\end{equation}
Now let $N$ be sufficiently large such that 
\begin{equation}
  \Bigg(1 - \left(1-\mathrm{BinoCDF}\big(n,p,\lceil(n+1)\alpha\rceil-1\big)\right)^N\Bigg) > p.
\end{equation}
Then
\begin{equation}
  \E\Big[ L_{n+1}\big(\lhat\big) \Big] > p^2
\end{equation}
For any $\alpha > 0$, we can take $p$ close enough to $1$ to render the claim false.
\end{proof}

\begin{proof}[Proof of Theorem~\ref{thm:monotonized}]
  Define the \emph{monotonized population risk} as 
  \begin{equation}
    \Rplus(\lambda) = \underset{t \geq \lambda}{\sup}\; \E\Big[ L_{n+1}(t) \Big]
  \end{equation}
  Note that the independence of $L_{n+1}$ and $\lhatplus_n$ implies that for all $n$,
  \begin{equation}
    \E\Big[ L_{n+1}\big( \lhatplus_n \big) \Big] \le \E\Big[ \Rplus\big( \lhatplus_n \big) \Big].
  \end{equation}
  Since $\Rplus$ is bounded, monotone, and one-dimensional, a generalization of the Glivenko-Cantelli Theorem given in Theorem 1 of~\cite{dehardt1971generalizations} gives that uniformly over $\lambda$, 
    \begin{equation}
      \underset{n \to \infty}{\lim} \sup_{\lambda}|\Rhat_{n}(\lambda) - R(\lambda)| \overset{a.s.}{\to} 0 .
  \end{equation}
  As a result,
  \begin{equation}
      \underset{n \to \infty}{\lim} \sup_{\lambda}|\Rhatplus_{n}(\lambda) - \Rplus(\lambda)| \overset{a.s.}{\to} 0,
  \end{equation}
  which implies that 
  \begin{equation}
      \underset{n \to \infty}{\lim} |\Rhatplus_{n}(\lhatplus) - \Rplus(\lhatplus)| \overset{a.s.}{\to} 0.
  \end{equation}
  By definition, $\Rhatplus(\lhatplus)\le \alpha$ almost surely and thus this directly implies
  \begin{equation}
      \underset{n \to \infty}{\limsup}~  \Rplus\big(\lhatplus_n\big) \leq \alpha\quad \text{a.s.}.
  \end{equation}
  Finally, since for all $n$, $\Rplus\big(\lhatplus_n\big) \leq B$, by Fatou's lemma,
  \begin{equation}
      \underset{n \to \infty}{\lim}\E\Big[ L_{n+1}\big( \lhatplus_n \big) \Big] \le \underset{n \to \infty}{\limsup}~\E\Big[ \Rplus\big( \lhatplus_n \big) \Big] \le \E\Big[ \underset{n \to \infty}{\limsup}~ \Rplus\big( \lhatplus_n \big) \Big] \leq \alpha.
  \end{equation}
\end{proof}

\begin{proof}[Proposition~\ref{thm:upper-bound-weighted}]
  Let 
  \begin{equation}
    \lhat' = \inf\left\{ \lambda : \frac{\sum_{i=1}^{n+1}w(X_i)L_i(\lambda)}{\sum_{i=1}^{n+1}w(X_i)} \leq \alpha \right\}.
  \end{equation}
   Since $\inf_{\lambda} L_i(\lambda) \le \alpha$, $\lhat'$ exists almost surely. Using the same argument as in the proof of Theorem \ref{thm:upper-bound}, we can show that $\lhat'\le \lhat(X_{n+1})$. Since $L_{n+1}(\lambda)$ is non-increasing in $\lambda$, 
  \[\E[L_{n+1}(\lhat(X_{n+1}))]\le \E[L_{n+1}(\lhat')].\]
  Let $E$ be the multiset of loss functions $\{ (X_1, Y_1), \ldots, (X_{n+1}, Y_{n+1}) \}$. Then $\lhat'$ is a function of $E$, or, equivalently, $\lhat'$ is a constant conditional on $E$. Lemma 3 of \cite{tibshirani2019conformal} implies that 
  \[(X_{n+1}, Y_{n+1})\mid E\sim \sum_{i=1}^{n+1}\frac{w(X_{i})}{\sum_{j=1}^{n+1}w(X_j)}\delta_{(X_j, Y_j)}\Longrightarrow L_{n+1} \mid E \sim \sum_{i=1}^{n+1}\frac{w(X_{i})}{\sum_{j=1}^{n+1}w(X_j)}\delta_{L_{i}}\]
 where $\delta_{z}$ denotes the Dirac measure at $z$. Together with the right-continuity of $L_{i}$, the above result implies
 \begin{equation}
    \E\left[L_{n+1}(\lhat')\mid E\right] = \frac{\sum_{i=1}^{n+1}w(X_i)L_{i}(\lhat')}{\sum_{i=1}^{n+1}w(X_i)} \leq \alpha.
 \end{equation}
 The proof is then completed by the law of total expectation.
\end{proof}

\begin{proof}[Proposition~\ref{thm:tv-bound}]
  Define the vector $Z'=(Z_1', \ldots, Z_n', Z_{n+1})$, where $Z_i' \overset{i.i.d.}{\sim} \mathcal{L}(Z_{n+1})$ for all $i \in [n]$. 
  Let 
  \[\eps = \sum_{i=1}^{n}\mathrm{TV}(Z_i, Z_i').\]
  By sublinearity,
  \begin{equation}
    \label{eq:tv-bounded-vector}
    \mathrm{TV}(Z, Z') \leq \epsilon.
  \end{equation}
  It is a standard fact that~\eqref{eq:tv-bounded-vector} implies
  \begin{equation}
      \underset{f \in \mathcal{F}_{\mathbbm{1}}}{\sup} \left| \E[ f(Z) ] - \E[ f(Z') ] \right| \leq \epsilon,
  \end{equation}
where $\mathcal{F}_{\mathbbm{1}} = \{f: \mathcal{Z}\mapsto [0, 1]\}$.
  Let $\ell : \mathcal{Z} \times \Lambda \to [0,B]$ be a bounded loss function.
  Furthermore, let $g(z) = \ell(z_{n+1}; \lhat(z_1, \ldots, z_n))$.
  Since $g(Z) \in [0, B]$, 
  \begin{equation}
      | \E[g(Z)] - \E[g(Z')] | \leq B\epsilon.
  \end{equation}
  Furthermore, since $Z_1', \ldots, Z_{n+1}'$ are exchangeable, we can apply Theorems~\ref{thm:upper-bound} and~\ref{thm:lower-bound} to $\E[g(Z')]$, recovering
  \begin{equation}
      \alpha - \frac{2B}{n+1} \leq \E[g(Z')] \leq \alpha.
  \end{equation}
  A final step of triangle inequality implies the result:
  \begin{equation}
      \alpha - \frac{2B}{n+1} - B\epsilon \leq \E[g(Z)] \leq \alpha + B\epsilon.
  \end{equation}
\end{proof}

\begin{proof}[Proposition~\ref{thm:quantile-risk-control}]
    It is left to prove that $\tilde{L}_i(\lambda)$ satisfies the conditions of Theorem \ref{thm:upper-bound}. It is clear that  $\tilde{L}_i(\lambda)\le 1$ and $\tilde{L}_i(\lambda)$ is non-increasing in $\lambda$ when $L_i(\lambda)$ is. Since $L_i(\lambda)$ is non-increasing and right-continuous, for any sequence $\lambda_{m}\downarrow \lambda$, 
    \[L_{i}(\lambda_{m})\uparrow L_i(\lambda)\Longrightarrow \ind{L_{i}(\lambda_{m}) > \alpha} \rightarrow \ind{L_{i}(\lambda) > \alpha}.\]
    Thus, $\tilde{L}_i(\lambda)$ is right-continuous. Finally, $L_i(\lambda_{\max}) \le \alpha$ implies  $\tilde{L}_i(\lambda_{\max}) = 0 \le 1 - \beta$. 
\end{proof}

\begin{proof}[Proposition~\ref{thm:multiple-risks}]
Examining~\eqref{eq:lhat-multiple-risks}, for each $\gamma \in \Gamma$, we have
\begin{equation}
    \E\left[L(\lhat, \gamma)\right] \leq \E\left[L(\lhat_{\gamma}, \gamma)\right] \leq \alpha(\gamma).
\end{equation}
Thus, dividing both sides by $\alpha(\gamma)$ and taking the supremum, we get that $\sup_{\gamma \in \Gamma} \E\left[\frac{L(\lhat, \gamma)}{\alpha(\gamma)}\right] \leq 1$, and the worst-case risk is controlled.
\end{proof}

\begin{proof}[Proposition~\ref{thm:adversarial}]
    Because $L_i(\lambda, \gamma)$ is bounded and monotone in $\lambda$ for all choices of $\gamma$, it is also true that $\tilde{L_i}(\lambda)$ is bounded and monotone.
    Furthermore, the pointwise supremum of right-continuous functions is also right-continuous.
    Therefore, the $\tilde{L_i}$ satisfy the assumptions of Theorem~\ref{thm:upper-bound}.
\end{proof}

\begin{proof}[Proposition~\ref{thm:U_risk_control}]
  Let 
  \[\lhat'_k = \inf\left\{\lambda: \frac{k!n!}{(n+k)!}\sum_{\S\subset \{1, \ldots, n+k\}, |\S| = k}L_{\S}(\lambda)\le \alpha\right\}.\]
  Since $L_\S(\lambda_{\max})\le \alpha$, $\lhat'_k$ exists almost surely. Since $L_{\S}(\lambda)\le B$, we have 
  \begin{align*}
    &\frac{k!n!}{(n+k)!}\sum_{\S\subset \{1, \ldots, n+k\}, |\S| = k}L_{\S}(\lambda)\\
    & \le \frac{k!n!}{(n+k)!}\sum_{\S\subset \{1, \ldots, n\}, |\S| = k}L_{\S}(\lambda) + B \cdot \sum_{\S\cap \{n+1, \ldots, n+k\}\neq \emptyset, |\S|=k}1\\
    & =   \frac{k!n!}{(n+k)!}\sum_{\S\subset \{1, \ldots, n\}, |\S| = k}L_{\S}(\lambda) + B \lb 1 - \frac{k!n!}{(n+k)!}\sum_{\S\subset \{1, \ldots, n\}, |\S| = k}1\rb\\
    & = \frac{k!n!}{(n+k)!}\sum_{\S\subset \{1, \ldots, n\}, |\S| = k}L_{\S}(\lambda) + B \lb 1 - \frac{(n!)^2}{(n+k)!(n-k)!}\rb.
  \end{align*}
  Since $L_{\S}(\lambda)$ is non-increasing in $\lambda$, we conclude that $\lhat'_k \le \lhat_k$ if the right-hand side of \eqref{eq:lhatk} is not empty; otherwise, by definition, $\lhat'_k \le \lambda_{\max} = \lhat_k$. Thus, $\lhat'_k \le \lhat_k$ almost surely. Let $E$ be the multiset of loss functions $\{L_{\S}: \S\subset \{1, \ldots, n+k\}, |\S| = k\}$. Using the same argument in the end of the proof of Theorem \ref{thm:upper-bound} and the right-continuity of $L_{\S}$, we can show that
  \[\E\left[L_{\{n+1, \ldots, n+k\}}(\lhat'_k)\mid E\right] = \frac{k!n!}{(n+k)!}\sum_{\S\subset \{1, \ldots, n+k\}, |\S| = k}L_{\S}(\lambda)\le \alpha.\]
  The proof is then completed by the law of iterated expectation.
  \end{proof}

\end{document}